\pdfoutput=1
\documentclass[11pt, letterpaper]{article}

\usepackage{fullpage}
\usepackage{layout}
\usepackage{multirow}

\usepackage{amsfonts}
\usepackage{amsmath}
\usepackage{amsthm}
\usepackage{amssymb} 

\usepackage{mdframed} 
\usepackage{thmtools}
\usepackage{enumitem}

\definecolor{shadecolor}{gray}{0.95}
\declaretheoremstyle[
headfont=\normalfont\bfseries,
notefont=\mdseries, notebraces={(}{)},
bodyfont=\normalfont,
postheadspace=0.5em,
spaceabove=1pt,
mdframed={
  skipabove=8pt,
  skipbelow=8pt,
  hidealllines=true,
  backgroundcolor={shadecolor},
  innerleftmargin=4pt,
  innerrightmargin=4pt}
]{shaded}

 \usepackage{xcolor}
\usepackage{color}
\usepackage{graphicx}

\usepackage{algorithm}
\usepackage[noend]{algpseudocode}
\usepackage{verbatim}

\newcommand{\R}{\mathbb{R}} 

\usepackage[colorinlistoftodos,bordercolor=orange,backgroundcolor=orange!20,linecolor=orange,textsize=scriptsize]{todonotes}

\newcommand{\eqdef}{\overset{\text{def}}{=}} 

\newcommand{\Exp}[1]{{\rm E}\left[#1\right] }    

\newcommand{\EE}[2]{{\bf E}_{#1}\left[#2\right] } 
\newcommand{\VV}[2]{{\bf Var}_{#1}\left[#2\right] }

\declaretheorem[style=shaded,within=section]{definition}
\declaretheorem[style=shaded,sibling=definition]{theorem}

\declaretheorem[style=shaded,sibling=definition]{lemma}

\theoremstyle{remark}
\newtheorem{example}{Example} 
\newtheorem{remark}{Remark} 
\usepackage[colorlinks=true,linkcolor=blue]{hyperref}

\title{Randomized  Distributed Mean Estimation: \\ Accuracy vs Communication}

 \author{Jakub Kone\v{c}n\'{y}\thanks{The author acknowledges support from Google via a Google European Doctoral Fellowship.} \qquad \qquad Peter Richt\'{a}rik\thanks{The author acknowledges support from Amazon, and the EPSRC Grant EP/K02325X/1, Accelerated Coordinate Descent Methods for Big Data Optimization.}  \\\\{\em School of Mathematics}\\{\em The University of Edinburgh}\\{\em United Kingdom}}

\begin{document}

\maketitle

\begin{abstract}
We consider the problem of estimating the arithmetic average of a finite collection of real vectors stored in a distributed fashion across several compute nodes subject to a communication budget constraint. Our analysis does not rely  on any statistical assumptions about the source of the vectors. This problem arises as a subproblem in many applications, including reduce-all operations within algorithms for distributed and federated optimization and learning. We propose a flexible family of randomized algorithms exploring the trade-off between expected communication cost and estimation error. Our family contains the full-communication and zero-error method on one extreme, and an $\epsilon$-bit communication and ${\cal O}\left(1/(\epsilon n)\right)$ error method on the opposite extreme. In the special case where we  communicate, in expectation, a single bit per coordinate of each vector, we improve upon existing results by obtaining $\mathcal{O}(r/n)$ error, where $r$ is the number of bits used to represent a floating point value. 
\end{abstract}

\section{Introduction} 
\label{sec:introduction}

We address the problem of estimating the arithmetic mean of $n$ vectors,  $X_1, \dots, X_n \in \R^d$, stored in a distributed fashion across $n$ compute nodes, subject to a constraint on the communication cost. 

In particular, we consider a star network topology  with a single server at the centre and $n$ nodes connected to it. All nodes send an encoded (possibly via a lossy randomized transformation) version of their vector  to the server, after which the server  performs a decoding operation to estimate the true mean \[X \eqdef \frac{1}{n}\sum_{i=1}^n X_i.\] The purpose of the encoding operation is to compress the vector so as to save on communication cost, which is typically the bottleneck in practical applications. 

To better illustrate the setup, consider the naive approach in which all nodes send the vectors without performing any encoding operation, followed by the application of a simple averaging decoder by the server. This results in zero estimation error at the expense of maximum communication cost of $ndr$ bits, where $r$ is the number  of bits needed to communicate a single floating point entry/coordinate  of $X_i$.

\subsection{Background and  Contributions}

The distributed mean estimation problem was recently  studied in a statistical framework where it is assumed that  the vectors $X_i$ are independent and identicaly distributed samples from some specific underlying distribution. In such a setup, the goal  is to estimate the true mean of the underlying distribution \cite{zhang2012communication, info_lower2013, comm_distr2014, braverman2015communication}. These works formulate lower and upper bounds on the communication cost  needed to achieve  the minimax optimal estimation error.

In contrast, we do not make any statistical assumptions on the source of the vectors, and study the trade-off between expected communication costs and mean square error of the estimate. Arguably, this setup is a more robust and accurate model of  the distributed mean estimation problems arising as subproblems in  applications such as reduce-all operations within algorithms for distributed and federated optimization \cite{HYDRA, add_vs_avg_2015, ma2015distributed, reddi2016aide, Federated_Opt}. In these applications, the averaging operations need to be done repeatedly throughout the iterations of a master learning/optimization algorithm, and  the vectors $\{X_i\}$  correspond to updates to a global model/variable. In these applications, the vectors evolve throughout the iterative process in a complicated  pattern, typically approaching zero as the master algorithm converges to optimality. Hence, their statistical properties change, which renders  fixed statistical assumptions not satisfied in practice.

For instance, when training a deep neural network model in a distributed environment, the vector $X_i$ corresponds to a stochastic gradient based on a minibatch of data stored on node $i$. In this setup we do not  have any useful prior statistical knowledge about the high-dimensional vectors to be aggregated. It has recently been observed that when communication cost is high, which is typically the case for commodity clusters, and even more so in a federated optimization framework, it is can be very useful to sacrifice on estimation accuracy in favor of reduced communication \cite{Federated_modelAVG, Federated_learning2016}.

In this paper we propose a {\em parametric family of randomized methods for estimating the mean $X$}, with parameters being a set of {\em probabilities} $p_{ij}$ for $i=1,\dots,n$ and $j=1,2,\dots,d$ and {\em node centers} $\mu_i \in \R^d$ for $i=1,2,\dots,n$. The exact meaning of these parameters is explained in Section~\ref{sec:encode}.   By varying the probabilities, at one extreme, we recover the exact  method described, enjoying zero estimation error at the expense of full communication cost. At the opposite extreme are methods with arbitrarily small expected communication cost, which is achieved at the expense of suffering an exploding estimation error. Practical methods appear somewhere on the continuum between these two extremes, depending on the specific requirements of the application at hand.  Suresh et al.\ \cite{Distributed_mean} propose a method combining a pre-processing step via a random structured rotation, followed by randomized binary quantization. Their quantization protocol arises as a suboptimal special case of our parametric family of methods.

To illustrate our results, consider the  special case in which we choose to communicate a single bit per element of $X_i$ only. We then obtain an $\mathcal{O}\left( \frac{r}{n}R \right)$ bound on the mean square error, where $r$ is number of bits used to represent a floating point value, and $R = \frac1n \sum_{i=1}^n \| X_i - \mu_i 1 \|^2$ with $\mu_i \in \R$ being the average of elements of $X_i$, and $1$ the all-ones vector in $\R^d$ (see Example 7 in Section~\ref{sec:examples}). Note that this bound improves upon the performance of the method of \cite{Distributed_mean} in two aspects. First, the bound is independent of $d$, improving from logarithmic dependence. Further, due to a preprocessing rotation step, their method requires $\mathcal{O}(d \log d)$ time to be implemented on each node, while our method is linear in  $d$. This and other special cases are summarized in Table~\ref{tbl:main1} in Section~\ref{sec:examples}. 

While the above already improves upon the state of the art, the improved results are in fact obtained for a suboptimal choice of the parameters of our method (constant probabilities $p_{ij}$, and node centers fixed to the mean $\mu_i$). One can decrease the MSE further by optimizing over the probabilities and/or node centers (see Section~\ref{sec:opt}). However, apart from a very low communication cost regime in which we have a closed form expression for the optimal probabilities, the problem needs to be solved numerically, and hence we do not have expressions for how much improvement is possible. We illustrate the effect of fixed and optimal probabilities on the trade-off between communication cost and MSE experimentally on a few selected datasets in Section~\ref{sec:opt} (see Figure~\ref{fig:uniform_vs_optimal}).

\subsection{Outline}

In Section~\ref{sec:3protocols} we formalize the concepts of encoding and decoding protocols. 
In Section~\ref{sec:encode} we describe a parametric family of randomized (and unbiased) encoding protocols and give a simple formula for the mean squared error. Subsequently, in Section~\ref{sec:comm} we formalize the notion of communication cost, and describe several communication protocols, which are optimal under different circumstances. We give simple instantiations of our protocol in Section~\ref{sec:examples}, illustrating the trade-off between communication costs and accuracy. In Section~\ref{sec:opt} we address the question of the optimal choice of parameters of our protocol. Finally, in Section~\ref{sec:further} we comment on possible extensions we leave out to future work.

\section{Three Protocols}
\label{sec:3protocols}

In this work we consider (randomized) {\em encoding protocols} $\alpha$, {\em communication protocols} $\beta$ and {\em decoding protocols} $\gamma$ using which the averaging is performed inexactly as follows. Node $i$ computes a (possibly stochastic) estimate of $X_i$ using the encoding protocol, which we denote $Y_i = \alpha(X_i) \in \R^d$, and sends it to the server using communication protocol $\beta$. By $\beta(Y_i)$ we denote the number of bits that need to be transferred under $\beta$. The server then estimates $X$ using the decoding protocol $\gamma$ of the estimates:
$$ Y \eqdef \gamma(Y_1, \dots, Y_n). $$

The objective of this work is to study the trade-off between the (expected) number of bits that need to be communicated, and the accuracy of $Y$ as an estimate of $X$. 

In this work we focus on encoders which are unbiased, in the following sense.

\begin{definition}[Unbiased and Independent Encoder] 
We say that encoder $\alpha$ is unbiased if  $\EE{\alpha}{\alpha(X_i)} = X_i$ for all $i=1,2,\dots,n$. We say that it is independent, if $\alpha(X_i)$ is independent from $\alpha(X_j)$ for all $i\neq j$.
\end{definition}

\begin{example}[Identity Encoder]
A trivial  example of an encoding protocol is the identity function: $\alpha(X_i) = X_i$.  It is both unbiased and independent. This encoder does not lead to any savings in communication that would be otherwise infeasible though.
\end{example}


We now formalize the notion of accuracy of estimating $X$ via $Y$. Since $Y$ can be random, the notion of accuracy will naturally be  probabilistic.

\begin{definition}[Estimation Error / Mean Squared Error] 
The {\em mean squared error} of protocol $(\alpha,\gamma)$ is the quantity 
\begin{eqnarray*} 
MSE_{\alpha,\gamma}(X_1, \dots, X_n) &=& \EE{\alpha,\gamma}{\| Y - X \|^2}\\
& =& \EE{\alpha, \gamma}{\left\| \gamma(\alpha(X_1),\dots,\alpha(X_n)) - X\right\|^2}.\end{eqnarray*}
\end{definition}

To illustrate the above concept, we now give a few examples:

\begin{example}[Averaging Decoder]
\label{ex:avg_decoder}
If $\gamma$ is the averaging function, i.e., $\gamma(Y_1, \dots, Y_n) = \frac1n \sum_{i=1}^n Y_i,$ then 
\[MSE_{\alpha,\gamma}(X_1, \dots, X_n) = \frac{1}{n^2}\EE{\alpha}{\left\| \sum_{i=1}^n \alpha(X_i) - X_i \right\|^2}.\]
\end{example}

The next example generalizes the identity encoder and averaging decoder.

\begin{example}[Linear Encoder and Inverse Linear Decoder]
\label{ex:linear_encoder}
Let $A:\R^d\to \R^d$ be linear and invertible.   Then we can set $Y_i = \alpha(X_i) \eqdef A X_i$ and $\gamma(Y_1,\dots,Y_n) \eqdef A^{-1} \left(\frac{1}{n}\sum_{i=1}^n Y_i\right)$. If $A$ is random, then $\alpha$ and $\gamma$ are random (e.g., a structured random rotation, see \cite{yu2016orthogonal}). Note that \[\gamma(Y_1,\dots,Y_n) =  \frac{1}{n}\sum_{i=1}^n A^{-1} Y_i = \frac{1}{n}\sum_{i=1}^n X_i = X,\]
and hence the MSE of $(\alpha,\gamma)$ is zero.
\end{example}

We shall now prove a simple result for unbiased and independent encoders used in subsequent sections.
\begin{lemma}[Unbiased and Independent Encoder + Averaging Decoder] 
\label{lem:general_MSE} 
If the encoder $\alpha$ is unbiased and independent, and $\gamma$ is the averaging decoder, then 
\[MSE_{\alpha,\gamma}(X_1, \dots, X_n) = \frac{1}{n^2}\sum_{i=1}^n \EE{\alpha}{\|Y_i-X_i\|^2} = \frac{1}{n^2}\sum_{i=1}^n \VV{\alpha}{\alpha(X_i)} .\]
\end{lemma}

\begin{proof} 
Note that $\EE{\alpha}{Y_i} = X_i$ for all $i$. We have
\begin{eqnarray}
MSE_{\alpha}(X_1,\dots,X_n) &=& \EE{\alpha}{\|Y-X\|^2} \notag \\
&\overset{(*)}{=}& \frac{1}{n^2} \EE{\alpha}{\left\|\sum_{i=1}^n Y_i - X_i \right\|^2} \notag \\
&\overset{(**)}{=}& \frac{1}{n^2} \sum_{i=1}^n \EE{\alpha}{\left\|Y_i - \EE{\alpha}{Y_i} \right\|^2} \notag \\
&=& \frac{1}{n^2} \sum_{i=1}^n \VV{\alpha}{\alpha(X_i)} \notag,
\end{eqnarray}
where (*) follows from unbiasedness and (**) from independence.
\end{proof}

One may wish to define the encoder as a combination of two or more separate encoders: $\alpha(X_i) = \alpha_2(\alpha_1(X_i))$. See \cite{Distributed_mean} for an example where $\alpha_1$ is a random rotation and $\alpha_2$ is binary quantization.

\section{A Family of Randomized Encoding Protocols}
\label{sec:encode}

Let $X_1,\dots, X_n\in \R^d$ be given. We shall write $X_{i} = (X_{i}(1),\dots, X_{i}(d))$ to denote the entries  of vector $X_i$. In addition, with each $i$ we also associate a parameter $\mu_i\in \R$. We refer to $\mu_i$ as the center of data at node $i$, or simply as {\em node center}. For now, we assume these parameters are fixed and we shall later comment on how to choose them optimally.

We shall define \emph{support} of $\alpha$ on node $i$ to be the set $S_i \eqdef \{ j \;:\; Y_i(j) \neq \mu_i \}$. We now define two parametric families of randomized encoding protocols. The first results in $S_i$ of random size, the second has  $S_i$ of a fixed size.

\subsection{Encoding Protocol with Variable-size Support}

With each pair $(i,j)$ we associate a parameter $0< p_{ij}\leq 1$,  representing a probability. The collection of parameters $\{p_{ij}, \mu_i\}$ defines an encoding protocol $\alpha$ as follows:
\begin{equation}
\label{eq:randomized_protocol}
Y_{i}(j) = 
\begin{cases}
\frac{X_{i}(j)}{p_{ij}} - \frac{1-p_{ij}}{p_{ij}} \mu_i & \quad \text{with probability} \quad p_{ij}, \\
\mu_i & \quad \text{with probability} \quad 1-p_{ij}.
\end{cases}
\end{equation}

\begin{remark}
Enforcing the probabilities to be positive, as opposed to nonnegative,  leads to vastly simplified notation in what follows. However, it is more natural to allow $p_{ij}$ to be  zero, in which case we have $Y_i(j)=\mu_i$ with probability 1. This raises issues such as potential lack of unbiasedness, which can be resolved, but only at the expense of a larger-than-reasonable notational overload.
\end{remark}


In the rest of this section, let $\gamma$ be the averaging decoder (Example~\ref{ex:avg_decoder}). Since $\gamma$ is fixed and deterministic, we shall for simplicity write $\EE{\alpha}{\cdot}$ instead of $\EE{\alpha,\gamma}{\cdot}$. Similarly, we shall write $MSE_\alpha(\cdot)$ instead of $MSE_{\alpha, \gamma}(\cdot)$.

We now prove two lemmas describing properties of the encoding protocol $\alpha$. Lemma~\ref{lem:unbiasedness} states that the protocol yields an unbiased estimate of the average $X$ and Lemma~\ref{lem:MSE} provides the expected mean square error of the estimate.

\begin{lemma}[Unbiasedness] 
\label{lem:unbiasedness}
The encoder $\alpha$ defined in \eqref{eq:randomized_protocol} is unbiased. That is,  $\EE{\alpha}{\alpha(X_i)} = X_i$ for all $i$. As a result, $Y$ is an unbiased estimate of the true average: $\EE{\alpha}{Y} = X$.
\end{lemma}

\begin{proof}
Due to linearity of expectation, it is enough to show that $\EE{\alpha}{Y(j)}=X(j)$ for all $j$. Since $Y(j) = \frac{1}{n}\sum_{i=1}^n Y_{i}(j)$ and $X(j) = \frac{1}{n}\sum_{i=1}^n X_{i}(j)$, it suffices to show that $\EE{\alpha}{Y_i(j)}=X_i(j)$:
$$ \EE{\alpha}{Y_i(j)} = p_{ij} \left( \frac{X_{i}(j)}{p_{ij}} - \frac{1-p_{ij}}{p_{ij}} \mu_i(j) \right) + (1-p_{ij}) \mu_i(j) = X_{i}(j), $$
and the claim is proved.
\end{proof}

\begin{lemma}[Mean Squared Error]
\label{lem:MSE}
Let $\alpha = \alpha(p_{ij},\mu_i)$ be the encoder defined in \eqref{eq:randomized_protocol}. Then
\begin{equation}
\label{eq:MSE_general}
MSE_{\alpha}(X_1, \dots, X_n) = \frac{1}{n^2} \sum_{i,j} \left( \frac{1}{p_{ij}} - 1 \right) \left( X_i(j) - \mu_i \right)^2.
\end{equation}
\end{lemma}

\begin{proof} 
Using Lemma~\ref{lem:general_MSE}, we have
\begin{eqnarray}
MSE_{\alpha}(X_1,\dots,X_n) &=&  \frac{1}{n^2} \sum_{i=1}^n \EE{\alpha}{\left\|Y_i - X_i \right\|^2} \notag \\
&=&\frac{1}{n^2} \sum_{i=1}^n \EE{\alpha}{\sum_{j=1}^d (Y_i(j) - X_i(j) )^2} \notag \\
&=&\frac{1}{n^2} \sum_{i=1}^n \sum_{j=1}^d \EE{\alpha}{ (Y_i(j) - X_i(j) )^2}.
\label{eq:6867sgs7}
\end{eqnarray}
For any $i,j$ we further have
\begin{align*}
\EE{\alpha}{ (Y_i(j) - X_i(j) )^2} &= 
p_{ij} \left( \frac{X_{i}(j)}{p_{ij}} - \frac{1-p_{ij}}{p_{ij}} \mu_i - X_i(j) \right)^2 + (1 - p_{ij}) \left( \mu_i - X_i(j) \right)^2 \\
&= \frac{(1 - p_{ij})^2}{p_{ij}} \left( X_i(j) - \mu_i \right)^2 + (1 - p_{ij}) \left( \mu_i - X_i(j) \right)^2 \\
&= \left( \frac{1-p_{ij}}{p_{ij}} \right) \left( X_i(j) - \mu_i \right)^2.
\end{align*}
It suffices to substitute the above into \eqref{eq:6867sgs7}.
\end{proof}

\subsection{Encoding Protocol with Fixed-size Support}

Here we propose an alternative encoding protocol, one with deterministic support size. As we shall see later, this results in deterministic communication cost.

Let $\sigma_k(d)$ denote the set of all subsets of $\{1, 2, \dots, d\}$ containing $k$ elements. The protocol $\alpha$ with a single integer parameter $k$ is then working as follows: First, each node $i$ samples $\mathcal{D}_i \in \sigma_k(d)$ uniformly at random, and then sets
\begin{equation}
\label{eq:randomized_protocol_2}
Y_{i}(j) = 
\begin{cases}
\frac{d X_{i}(j)}{k} - \frac{d-k}{k} \mu_i & \quad \text{if} \quad j \in \mathcal{D}_i, \\
\mu_i & \quad \text{otherwise}.
\end{cases}
\end{equation}

Note that due to the design, the size of the support of $Y_i$ is always $k$, i.e., $|S_i| = k$. Naturally, we can expect this protocol to perform practically the same as the protocol \eqref{eq:randomized_protocol} with $p_{ij} = k/d$, for all $i, j$. Lemma~\ref{lem:MSE_2} indeed suggests this is the case. While this protocol admits  a more efficient communication protocol (as we shall see in Section~),  protocol \eqref{eq:randomized_protocol} enjoys a larger parameters space, ultimately leading to better MSE.  We comment on this tradeoff in subsequent sections.

As for the data-dependent protocol, we prove basic properties. The proofs are similar to those of Lemmas~\ref{lem:unbiasedness} and \ref{lem:MSE} and we defer them to Appendix~\ref{sec:app:alternative_protocol}.

\begin{lemma}[Unbiasedness] 
\label{lem:unbiasedness_2}
The encoder $\alpha$ defined in \eqref{eq:randomized_protocol} is unbiased. That is,  $\EE{\alpha}{\alpha(X_i)} = X_i$ for all $i$. As a result, $Y$ is an unbiased estimate of the true average: $\EE{\alpha}{Y} = X$.
\end{lemma}

\begin{lemma}[Mean Squared Error]
\label{lem:MSE_2} 
Let $\alpha = \alpha(k)$ be encoder defined as in \eqref{eq:randomized_protocol_2}. Then
\begin{equation}
\label{eq:MSE_general_2}
MSE_{\alpha}(X_1, \dots, X_n) = \frac{1}{n^2} \sum_{i=1}^n \sum_{j=1}^d \left( \frac{d-k}{k} \right) \left( X_i(j) - \mu_i \right)^2.
\end{equation}
\end{lemma}

\section{Communication Protocols}
\label{sec:comm}

Having defined the encoding protocols $\alpha$, we need to specify the way the encoded vectors $Y_i = \alpha(X_i)$, for $i=1,2,\dots,n$, are communicated to the server. Given a specific  {\em communication protocol} $\beta$, we write $\beta(Y_i)$ to denote the (expected) number of bits that are communicated by node $i$ to the server. Since $Y_i = \alpha(X_i)$ is in general not deterministic, $\beta(Y_i)$ can be a random variable.

\begin{definition}[Communication Cost] 
The {\em communication cost} of communication protocol $\beta$ under randomized encoding $\alpha$ is the total expected number of bits transmitted to the server: 
\begin{equation}
\label{eq:989d8dd} 
C_{\alpha,\beta}(X_1,\dots,X_n) = \EE{\alpha}{\sum_{i=1}^n  \beta(\alpha(X_i))}. 
\end{equation}
\end{definition}

Given $Y_i$, a good communication protocol is able to encode $Y_i=\alpha(X_i)$ using a few bits only.  Let $r$ denote the number of bits used to represent a floating point number.  Let $\bar{r}$ be the  the number of bits representing $\mu_i$.

In the rest of this section we describe several communication protocols $\beta$ and calculate their communication cost.

\subsection{Naive} 
Represent $Y_i=\alpha(X_i)$ as $d$ floating point numbers. Then for all encoding protocols $\alpha$ and all $i$ we have $\beta(\alpha(X_i)) = dr$, whence \[C_{\alpha,\beta} = \EE{\alpha}{\sum_{i=1}^n \beta(\alpha(X_i))} = ndr.\]

\subsection{Varying-length}
We will use a single variable for every element of the vector $Y_i$, which does not have constant size. The first bit decides whether the value represents $\mu_i$ or not. If yes, end of variable, if not, next $r$ bits represent the value of $Y_i(j)$. In addition, we need to communicate $\mu_i$, which takes $\bar r$ bits\footnote{The distinction here is because $\mu_i$ can be chosen to be data independent, such as $0$, so we don't have to communicate anything (i.e., $\bar r = 0$)}. We thus have
\begin{equation}
\label{eq:s09y09hjfff}
\beta(\alpha(X_i)) = \bar r + \sum_{j=1}^d \left(1_{(Y_i(j) = \mu_i)} + (r+1) \times 1_{(Y_i(j) \neq \mu_i)}\right),
\end{equation}
where $1_{e}$ is the indicator function of event $e$. The expected number of bits communicated is given by

\begin{align*} 
\label{eq:9y80s9y0sd}
C_{\alpha,\beta} = \EE{\alpha}{\sum_{i=1}^n\beta(\alpha(X_i)))} &\overset{\eqref{eq:s09y09hjfff}}{=} n \bar r +\sum_{i=1}^n \sum_{j=1}^d \left(1-p_{ij} + (r+1) p_{ij} \right) \\
&= n \bar r +\sum_{i=1}^n \sum_{j=1}^d \left(1 + r p_{ij} \right)
\end{align*}
In the special case when $p_{ij} = p>0 $ for all $i,j$, we get
$$ C_{\alpha,\beta} = n(\bar{r} + d + pdr). $$

\subsection{Sparse Communication Protocol for Encoder~\eqref{eq:randomized_protocol}} 
We can represent $Y_i$ as a sparse vector; that is, a list of pairs $(j, Y_i(j))$ for which $Y_i(j) \neq \mu_i$. The number of bits to represent each pair is $\lceil \log(d)\rceil + r$. Any index not found in the list, will be interpreted by server as having value $\mu_i$. Additionally, we have to communicate the value of $\mu_i$ to the server, which takes $\bar r$ bits. We assume that the value $d$, size of the vectors, is known to the server. Hence,
\[\beta(\alpha(X_i)) = \bar r + \sum_{j=1}^d 1_{(Y_i(j) \neq \mu_i)} \times \left( \lceil \log d \rceil + r \right) . \]
Summing up through $i$ and  taking expectations, the the communication cost is given by
\begin{equation}
\label{eq:s09y09y9ff}
C_{\alpha,\beta} = \EE{\alpha}{\sum_{i=1}^n\beta(\alpha(X_i))} = n \bar{r} + (\lceil \log d \rceil + r)  \sum_{i=1}^n \sum_{j=1}^d p_{ij}.
\end{equation}
In the special case when $p_{ij} = p>0 $ for all $i,j$, we get
$$ C_{\alpha,\beta} = n\bar{r} + (\lceil \log d \rceil + r) ndp. $$

\begin{remark}
A practical improvement upon this could be to (without loss of generality) assume that the pairs $(j, Y_i(j))$ are ordered by $j$, i.e., we have $\{ (j_s, Y_i(j_s)) \}_{s=1}^k$ for some $k$ and $j_1 < j_2 < \dots < j_k$. Further, let us denote $j_0 = 0$. We can then use a variant of variable-length quantity \cite{wikiVLQ} to represent the set $\{ (j_s - j_{s-1}, Y_i(j_s)) \}_{s=1}^k$. With careful design one can hope to reduce the $\log(d)$ factor in the average case. Nevertheless, this does not improve the worst case analysis we focus on in this paper, and hence we do not delve deeper in this.
\end{remark}

\subsection{Sparse Communication Protocol for Encoder \eqref{eq:randomized_protocol_2}} 

We now describe a sparse  communication protocol compatible only with fixed length encoder defined in \eqref{eq:randomized_protocol_2}. Note that subset selection can  be compressed in the form of a random seed, letting us avoid the $\log(d)$ factor in \eqref{eq:s09y09y9ff}. This includes the protocol defined in \eqref{eq:randomized_protocol_2} but also \eqref{eq:randomized_protocol} with uniform probabilities $p_{ij}$.

In particular, we can represent $Y_i$ as a sparse vector containing the list of the values for which $Y_i(j) \neq \mu_i$, ordered by $j$. Additionally, we need to communicate  the value  $\mu_i$ (using $\bar r$ bits) and a random seed (using $\bar r_s$ bits), which can be used to reconstruct the indices $j$, corresponding to the communicated values. Note that for any fixed $k$ defining protocol \eqref{eq:randomized_protocol_2}, we have $|S_i|=k$. Hence, communication cost is deterministic:
\begin{equation}
\label{eq:beta_deterministic_sparse}
C_{\alpha, \beta} = \sum_{i=1}^n \beta(\alpha(X_i)) = n (\bar r + \bar r_s) + nkr.
\end{equation}

In the case of the variable-size-support encoding protocol \eqref{eq:randomized_protocol} with $p_{ij} = p > 0$ for all $i, j$, the  sparse communication protocol described here yields expected communication cost
\begin{equation}
\label{eq:exp_beta_deterministic_sparse}
C_{\alpha,\beta} = \EE{\alpha}{\sum_{i=1}^n\beta(\alpha(X_i))} = n (\bar r + \bar r_s) + ndpr.
\end{equation}

\subsection{Binary}
If the elements of $Y_i$ take only two different values, $Y_i^{min}$ or $Y_i^{max}$, we can use a {\em binary communication protocol}. That is, for each node $i$, we communicate the values of $Y_i^{min}$ and $Y_i^{max}$ (using $2r$ bits), followed by a single bit per element of the array indicating whether $Y_i^{max}$ or $Y_i^{min}$ should be used. The resulting (deterministic) communication cost is 
\begin{equation}
\label{eq:binary_comm_cost}
C_{\alpha, \beta} = \sum_{i=1}^n \beta(\alpha(X_i)) = n (2 r) + nd.
\end{equation}

\subsection{Discussion}
In the above, we have presented several communication protocols of different complexity. However, it is not possible to claim any of them is the most efficient one. Which communication protocol is the best, depends on the specifics of the used encoding protocol. Consider the extreme case of encoding protocol \eqref{eq:randomized_protocol} with $p_{ij} = 1$ for all $i, j$. The naive communication protocol is clearly the most efficient, as all other protocols need to send some additional information.

However, in the interesting case when we consider small communication budget, the sparse communication protocols are the most efficient. Therefore, in the following sections, we focus primarily on optimizing the performance using these protocols.

\section{Examples}
\label{sec:examples}

In this section, we highlight on  several instantiations of our protocols, recovering existing techniques and formulating novel ones. We comment on the resulting trade-offs between  communication cost and estimation error.

\subsection{Binary Quantization}

We start by recovering an existing method, which turns every element of the vectors $X_i$ into a particular binary representation.

\begin{example}
\label{ex:suresh}
If we set the parameters of protocol \eqref{eq:randomized_protocol} as $\mu_i = X_i^{min}$ and $p_{ij} = \frac{X_i(j) - X_i^{min}}{\Delta_i}$, where $\Delta_i\eqdef X_i^{max} - X_i^{min}$ (assume, for simplicity, that $\Delta_i \neq 0$), we exactly recover the  quantization algorithm proposed in \cite{Distributed_mean}: 
\begin{equation}
\label{eq:protocol_special_case_felix}
Y_{i}(j) = 
\begin{cases}
X_i^{max} & \quad \text{with probability} \quad \frac{X_i(j) - X_i^{min}}{\Delta_i}, \\
X_i^{min} & \quad \text{with probability} \quad \frac{X_i^{max} - X_i(j)}{\Delta_i}.
\end{cases}
\end{equation}

Using the formula \eqref{eq:MSE_general} for the encoding protocol $\alpha$, we get
\begin{align*}
MSE_\alpha &= \frac{1}{n^2} \sum_{i=1}^n \sum_{j=1}^d \frac{X_i^{max} - X_i(j)}{X_i(j) - X_i^{min}} \left( X_i(j) - X_i^{min} \right)^2 \leq \frac{d}{2 n} \cdot \frac1n \sum_{i=1}^n \|X_i\|^2.
\end{align*}
This exactly recovers the MSE bound established in \cite[Theorem 1]{Distributed_mean}. Using the binary communication protocol yields the communication cost of $1$ bit per element if $X_i$, plus a two real-valued scalars \eqref{eq:binary_comm_cost}.
\end{example}

\begin{remark}
\label{rem:rotation_quantization_protocol}
If we use the above protocol jointly with randomized linear encoder and decoder (see Example~\ref{ex:linear_encoder}), where the linear transform is the randomized Hadamard transform, we recover the method described in \cite[Section 3]{Distributed_mean} which yields improved $MSE_\alpha = \frac{2\log d + 2}{n} \cdot \frac1n \sum_{i=1}^n \|X_i\|^2$ and can be implemented in $\mathcal{O} (d \log d)$ time.
\end{remark}

\subsection{Sparse Communication Protocols}

Now we move to comparing the communication costs and estimation error of various instantiations of the encoding protocols, utilizing the deterministic sparse communication protocol and uniform probabilities.

For the remainder of this section, let us only consider instantiations of our protocol where $p_{ij} = p > 0$ for all $i, j$, and assume that the node centers are set to the vector averages, i.e., $\mu_i = \frac1d \sum_{j=1}^d X_i(j)$. Denote $R = \frac1n \sum_{i=1}^n \sum_{j=1}^d (X_i(j) - \mu_i)^2$. For simplicity, we also assume that $|S| = nd$, which is what we can in general expect without any prior knowledge about the vectors $X_i$. 

The properties of the following examples follow from Equations~\eqref{eq:MSE_general} and \eqref{eq:exp_beta_deterministic_sparse}. When considering the communication costs of the protocols, keep in mind that the trivial benchmark is $C_{\alpha, \beta} = ndr$, which is achieved by simply sending the vectors unmodified. Communication cost of $C_{\alpha, \beta} = nd$ corresponds to the interesting special case when we use (on average) one bit per element of each $X_i$.

\begin{example}[Full communication]
\label{ex:full}
If we choose $p = 1$, we get
$$ C_{\alpha, \beta} = n(\bar r_s + \bar r) + ndr, \qquad MSE_{\alpha, \gamma} = 0. $$
In this case, the encoding protocol is lossless, which ensures $MSE = 0$. Note that in this case, we could get rid of the $n(\bar r_s + \bar r)$ factor by using naive communication protocol.
\end{example}

\begin{example}[Log MSE]
\label{ex:log_MSE}
If we choose $p = 1 / \log d$, we get
$$ C_{\alpha, \beta} = n(\bar r_s + \bar r) + \frac{ndr}{\log d}, \qquad MSE_{\alpha, \gamma} = \frac{\log(d) - 1}{n} R. $$
This protocol order-wise matches the $MSE$ of the method in Remark~\ref{rem:rotation_quantization_protocol}. However, as long as $d > 2^r$, this protocol attains this error with \emph{smaller} communication cost. In particular, this is on expectation \emph{less} than a single bit per element of $X_i$. Finally, note that the factor $R$ is always smaller or equal to the factor $\frac1n \sum_{i=1}^n \| X_i \|^2$ appearing in Remark~\ref{rem:rotation_quantization_protocol}.
\end{example}

\begin{example}[1-bit per element communication]
\label{ex:1_bit}
If we choose $p = 1 / r$, we get
$$ C_{\alpha, \beta} = n(\bar r_s + \bar r) + nd, \qquad MSE_{\alpha, \gamma} = \frac{r - 1}{n} R. $$
This protocol communicates on expectation single bit per element of $X_i$ (plus additional $\bar r_s + \bar r$ bits per client), while attaining bound on $MSE$ of $\mathcal{O}(r/n)$. To the best of out knowledge, this is the first method to attain this bound without additional assumptions. \end{example}

\begin{example}[Alternative 1-bit per element communication]
\label{ex:1_bit_alternative}
If we choose $p = \frac{d - \bar r_s - \bar r}{dr}$, we get
$$ C_{\alpha, \beta} =  nd, \qquad MSE_{\alpha, \gamma} = \frac{\frac{dr}{d - \bar r_s - \bar r} - 1}{n} R. $$
This alternative protocol attains on expectation exactly single bit per element of $X_i$, with (a slightly more complicated) $\mathcal{O} (r/n)$ bound on $MSE$.
\end{example}

\begin{example}[Below 1-bit communication]
\label{ex:below_1_bit}
If we choose $p = 1 / d$, we get
$$ C_{\alpha, \beta} = n(\bar r_s + \bar r) + nr, \qquad MSE_{\alpha, \gamma} = \frac{d - 1}{n} R. $$
This protocol attains the MSE of protocol in Example~\ref{ex:suresh} while at the same time communicating on average significantly less than a single bit per element of $X_i$.
\end{example}

\begin{table}
\begin{center}
\begin{tabular}{|c|c|c|c|}
\hline
Example & $p$ & $C_{\alpha, \beta}$ & $MSE_{\alpha, \gamma}$ \\
\hline
Example \ref{ex:full} (Full) & $1$ & $ndr$ & $0$ \\
Example \ref{ex:log_MSE} (Log $MSE$) & $1 / \log d$ & $n(\bar r_s + \bar r) + \frac{ndr}{\log d}$ & $(\log(d) - 1) \tfrac{R}{n}$ \\
Example \ref{ex:1_bit} ($1$-bit) & $1 / r$ & $n(\bar r_s + \bar r) + nd$ & $(r - 1) \tfrac{R}{n}$ \\
Example \ref{ex:below_1_bit} (below $1$-bit) & $1 / d$ & $n(\bar r_s + \bar r) + nr$ & $(d - 1) \tfrac{R}{n}$ \\
\hline
\end{tabular}
\end{center}
\caption{Summary of achievable communication cost and estimation error, for various choices of probability $p$.}
\label{tbl:main1}
\end{table}

We summarize these examples in Table~\ref{tbl:main1}.

Using the deterministic sparse protocol, there is an obvious lower bound on the communication cost --- $n(\bar r_s + \bar r)$. We can bypass this threshold by using the sparse protocol, with a data-independent choice of $\mu_i$, such as $0$, setting $\bar r = 0$. By setting $p = \epsilon / d( \lceil \log d \rceil + r)$, we get arbitrarily small expected communication cost of $C_{\alpha, \beta} = \epsilon$, and the cost of exploding estimation error $MSE_{\alpha, \gamma} = \mathcal{O}(1 / \epsilon n)$.

Note that all of the above examples have random communication costs. What we present is the \emph{expected} communication cost of the protocols. All the above examples can be modified to use the encoding protocol with fixed-size support defined in $\eqref{eq:randomized_protocol_2}$ with the parameter $k$ set to the value of $pd$ for corresponding $p$ used above, to get the same results. The only practical difference is that the communication cost will be deterministic for each node, which can be useful for certain applications.

\section{Optimal Encoders} \label{sec:opt}

Here we consider $(\alpha, \beta, \gamma)$, where $\alpha=\alpha(p_{ij},\mu_i)$ is the encoder defined in \eqref{eq:randomized_protocol}, $\beta$ is the associated the sparse communication protocol, and $\gamma$ is the averaging decoder.  Recall from Lemma~\ref{eq:MSE_general} and \eqref{eq:s09y09y9ff}
that the mean square error and communication cost are given by:
\begin{equation}
\label{eq:MSE+C}
MSE_{\alpha,\gamma} = \frac{1}{n^2} \sum_{i,j} \left( \frac{1}{p_{ij}} - 1 \right) \left( X_i(j) - \mu_i \right)^2 , \quad C_{\alpha,\beta} =  n \bar{r} + (\lceil \log d \rceil + r)  \sum_{i=1}^n \sum_{j=1}^d p_{ij}.
\end{equation}
 
Having these closed-form formulae as functions of the parameters $\{p_{ij}, \mu_i\}$, we can now ask questions such as:
\begin{enumerate}
\item Given a communication budget, which encoding protocol has the smallest mean squared error?
\item Given a bound on the mean squared error, which encoder  suffers the minimal communication cost? 
\end{enumerate}
Let us now address the first question; the second question can  be handled in a similar fashion.  In particular, consider the optimization problem
\begin{eqnarray}
\text{minimize}& & \sum_{i,j} \left(\frac{1}{p_{ij}} -1\right) (X_i(j)-\mu_i)^2 \notag \\
\text{subject to} && \mu_i \in \R, \quad i=1,2,\dots,n \notag\\
&& \sum_{i,j} p_{ij} \leq B \label{eq:optimalCCyyy}\\
&& 0< p_{ij} \leq 1, \quad i=1,2,\dots,n; \quad j=1,2,\dots,d,
\end{eqnarray}
where $B>0$ represents a  bound on the part of the total communication cost in \eqref{eq:MSE+C} which depends on the choice of the probabilities $p_{ij}$.

Note that while the constraints in \eqref{eq:optimalCCyyy} are convex (they are linear), the objective is not jointly convex in $\{p_{ij},\mu_i\}$. However, the objective is convex in $\{p_{ij}\}$ and convex in $\{\mu_i\}$. This suggests a simple {\em alternating minimization} heuristic for solving the above problem: 
\begin{enumerate}
\item  Fix the probabilities and optimize over the node centers,
 \item Fix the node centers and optimize over probabilities.
\end{enumerate}

These two steps are  repeated until a suitable convergence criterion is reached. Note that the first step has a closed form solution. Indeed, the problem decomposes across the node centers to $n$ univariate unconstrained convex quadratic minimization problems, and the solution is given by
\begin{equation}\label{eq:opt_node_centers_given_prob}\mu_i = \frac{\sum_j w_{ij} X_i(j)}{\sum_{j} w_{ij}},\qquad w_{ij}\eqdef \frac{1}{p_{ij}}-1.\end{equation}
The second step does not  have a closed form solution in general; we provide an analysis of  this step in Section~\ref{sec:0y09fyhff}.

\begin{remark}Note that the upper bound $\sum_{i,j}(X_i(j)-\mu_i)^2 / p_{ij}$ on the objective is jointly convex in $\{p_{ij},\mu_i\}$. We may therefore instead optimize this upper bound by a suitable convex optimization algorithm.
\end{remark}

\begin{remark} An alternative and a more practical model to \eqref{eq:optimalCCyyy} is to choose per-node budgets $B_1,\dots,B_n$ and require $\sum_{j} p_{ij} \leq B_i$ for all $i$. The problem becomes separable across the nodes, and can therefore be solved by each node independently. If we set $B=\sum_i B_i$, the optimal solution obtained this way will lead to MSE which is lower bpunded by the MSE obtained through  \eqref{eq:optimalCCyyy}.
\end{remark}

\subsection{Optimal Probabilities for Fixed Node Centers}\label{sec:0y09fyhff}

Let the node centers $\mu_i$ be fixed. Problem \eqref{eq:optimalCCyyy} (or, equivalently, step 2 of the alternating minimization method described above)  then takes the form 
\begin{eqnarray}
\text{minimize}& & \sum_{i,j} \frac{(X_i(j)-\mu_i)^2}{p_{ij}} \notag \\
\text{subject to} && \sum_{i,j} p_{ij} \leq B \label{eq:optimalCCxxx}\\
&& 0< p_{ij} \leq 1, \quad  i=1,2,\dots n, \quad j=1,2,\dots,d \notag.
\end{eqnarray}

Let $S = \{(i,j)\;:\; X_i(j) \neq \mu_i\}$. Notice that as long as $B\geq |S|$, the optimal solution is to set $p_{ij}=1$ for all $(i,j)\in S$ and $p_{ij}= 0$ for all $(i,j)\notin S$.\footnote{We interpret $0/0$ as $0$ and do not worry about infeasibility. These issues can be properly formalized by allowing $p_{ij}$ to be zero in the encoding protocol and in \eqref{eq:optimalCCxxx}. However, handling this singular situation requires a notational overload which we are not willing to pay.} In such a case, we have $MSE_{\alpha,\gamma} = 0.$ Hence, we can without loss of generality assume that $B \leq |S|$.

While we are not able to derive a closed-form solution to this problem, we can formulate upper and lower bounds on the optimal estimation error, given a bound on the communication cost formulated via $B$.

\begin{theorem}[MSE-Optimal Protocols subject to a Communication Budget]
\label{thm:main}
Consider problem \eqref{eq:optimalCCxxx} and fix any $B\leq |S|$. Using the sparse communication protocol $\beta$, the optimal encoding protocol $\alpha$ has communication complexity \begin{equation}\label{eq:9sy09yhiffs}C_{\alpha,\beta} = n\bar r + (\lceil \log d\rceil + r )B,\end{equation} and the mean squared error satisfies the bounds
\begin{equation}\label{eq:s0y09yhff} \left(\frac{1}{B}-1\right)\frac{R}{n} \leq MSE_{\alpha,\gamma} \leq \left(\frac{|S|}{B}-1\right)\frac{R}{n},\end{equation}
where $R = \frac{1}{n}\sum_{i=1}^n\sum_{j=1}^d (X_i(j)-\mu_i)^2 = \frac{1}{n}\sum_{i=1}^n \|X_i - \mu_i 1\|^2$.
Let $a_{ij}=|X_{i}(j)-\mu_i|$ and $W=\sum_{i,j} a_{ij}$. If, moreover, $B \leq \sum_{(i,j)\in S} a_{ij}/ \max_{(i,j)\in S} a_{ij}$ (which is true, for instance, in the ultra-low communication regime with $B\leq 1$), then 
\begin{equation}\label{eq:s08y09fh9ff}MSE_{\alpha,\gamma}  = \frac{W^2}{n^2 B}- \frac{R}{n}.\end{equation}
\end{theorem}
\begin{proof}
Setting $p_{ij} = B/|S|$ for all $(i,j)\in S$ leads to a feasible solution of \eqref{eq:optimalCCxxx}. In view of \eqref{eq:MSE+C}, one then has 
\[MSE_{\alpha,\gamma} = \frac{1}{n^2} \left(\frac{|S|}{B}-1\right)\sum_{(i,j)\in S}  \left( X_i(j) - \mu_i \right)^2 =  \left(\frac{|S|}{B}-1\right) \frac{R}{n},\]
where 
$R = \frac{1}{n}\sum_{i=1}^n\sum_{j=1}^d (X_i(j)-\mu_i)^2 = \frac{1}{n}\sum_{i=1}^n \|X_i - \mu_i 1\|^2$.

If we relax the problem by removing the constraints $p_{ij}\leq 1$, the optimal solution satisfies  $a_{ij}/p_{ij} = \theta>0$  for all $(i,j)\in S$.  At optimality the bound involving $B$ must be tight, which leads to
$\sum_{(i,j)\in S} a_{ij}/\theta = B$, whence $\theta = \tfrac{1}{B}\sum_{(i,j)\in S} a_{ij}$. So, $p_{ij} = a_{ij}B/\sum_{(i,j)\in S}a_{ij}$. The optimal MSE therefore satisfies the lower bound
\[MSE_{\alpha,\gamma} \geq \frac{1}{n^2} \sum_{(i,j)\in S} \left(\frac{1}{p_{ij}}-1\right) \left( X_i(j) - \mu_i \right)^2 =  \frac{1}{n^2 B}W^2 - \frac{R}{n},\]
where $W \eqdef \sum_{(i,j)\in S} a_{ij} \geq \left(\sum_{(i,j)\in S} a_{ij}^2\right)^{1/2} = (nR)^{1/2}$. Therefore,
$MSE_{\alpha,\gamma} \geq \left(\frac{1}{B}-1\right) \frac{R}{n}$. If $B \leq \sum_{(i,j)\in S}a_{ij} / \max_{(i,j)\in S} a_{ij}$, then $p_{ij}\leq 1$ for all $(i,j)\in S$, and hence we have optimality. (Also note that, by Cauchy-Schwarz inequality,  $W^2\leq n R |S|$.)
\end{proof}


\subsection{Trade-off Curves}

To illustrate the trade-offs between communication cost and estimation error (MSE) achievable by the protocols discussed in this section, we present simple numerical examples  in Figure~\ref{fig:uniform_vs_optimal}, on three synthetic data sets with  $n=16$ and $d=512$. We choose an array of values  for $B$, directly bounding the communication cost via \eqref{eq:9sy09yhiffs}, and evaluate the $MSE$ \eqref{eq:MSE_general} for three encoding protocols (we use the sparse communication protocol and averaging decoder). All these protocols have the same communication cost, and only differ in the selection of the parameters $p_{ij}$ and $\mu_i$. In particular, we consider 
\begin{itemize}
\item[(i)] uniform probabilities $p_{ij}=p>0$ with average node centers $\mu_i = \frac{1}{d}\sum_{j=1}^d X_i(j)$ (blue dashed line), 
\item[(ii)]  optimal probabilities $p_{ij}$ with average node centers $\mu_i = \frac{1}{d}\sum_{j=1}^d X_i(j)$ (green dotted line), and 
\item[(iii)] optimal probabilities with optimal node centers,  obtained via the alternating minimization approach described above (red solid line).
\end{itemize}

In order to put a scale on the horizontal axis, we assumed that $r = 16$. Note that, in practice, one would choose $r$ to be as small as possible without adversely affecting  the application utilizing our  distributed mean estimation method. The three plots represent $X_i$ with entries drawn in an i.i.d.\ fashion from Gaussian ($\mathcal{N}(0, 1)$), Laplace ($\mathcal{L}(0, 1)$) and chi-squared ($\chi^2(2)$) distributions, respectively. As we can see, in the case of non-symmetric distributions, it is not necessarily optimal to set the node centers to  averages. 

As expected, for fixed node centers, optimizing over probabilities results in improved performance, across the entire trade-off curve. That is, the curve shifts downwards. In the first two plots based on data from symmetric distributions (Gaussian and Laplace), the average node centers are nearly optimal,  which explains why the red solid and green dotted lines coalesce. This can be also established formally. In the third plot, based on the non-symmetric chi-squared data, optimizing over node centers leads to further improvement, which gets more pronounced with increased communication budget. It is possible to generate data where the difference between any pair of the three trade-off curves becomes arbitrarily large. 

Finally, the black cross represents performance of the quantization protocol from Example~\ref{ex:suresh}. This approach appears as a single point in the trade-off space due to lack of any parameters to be fine-tuned.

\begin{figure}
\centering
\includegraphics[width=0.32\linewidth]{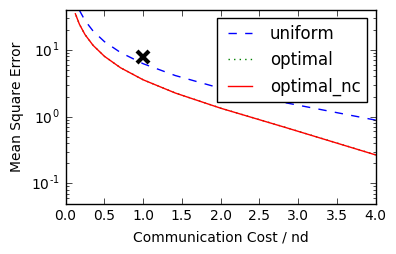}
\includegraphics[width=0.32\linewidth]{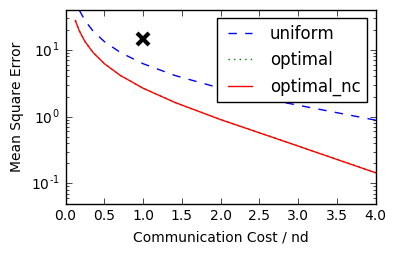}
\includegraphics[width=0.32\linewidth]{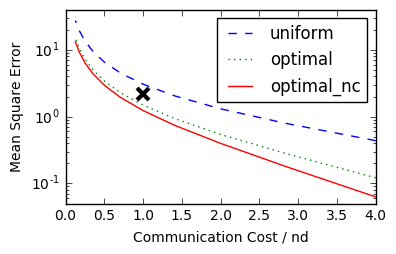}
\caption{{\em Trade-off curves} between communication cost and estimation error (MSE) for four protocols. The plots correspond to vectors $X_i$ drawn in an i.i.d.\ fashion from Gaussian, Laplace and $\chi^2$ distributions, from left to right. The black cross marks the performance of binary quantization (Example~\ref{ex:suresh}).}
\label{fig:uniform_vs_optimal}
\end{figure}

%
%
%

\section{Further Considerations} \label{sec:further}

In this section we outline further ideas worth consideration. However, we leave a detailed analysis to future work.

\subsection{Beyond Binary Encoders}

We can generalize the binary encoding protocol~\eqref{eq:randomized_protocol} to a $k$-ary protocol. To illustrate the concept without unnecessary notation overload, we present only the ternary (i.e., $k=3$) case.

Let the collection of parameters $\{p'_{ij}, p''_{ij}, \bar{X}'_i, \bar{X}''_i\}$ define an encoding protocol $\alpha$ as follows:

\begin{equation}
\label{eq:randomized_protocol_3}
Y_{i}(j) = 
\begin{cases}
\bar X'_i & \quad \text{with probability} \quad p'_{ij}, \\
\bar X''_i & \quad \text{with probability} \quad p''_{ij}, \\
\frac{1}{1 - p'_{ij} - p''_{ij}} \left( X_i(j) - p'_{ij} \bar X'_i - p''_{ij} \bar X''_i \right) & \quad \text{with probability} \quad 1 - p'_{ij} - p''_{ij}.
\end{cases}
\end{equation}

It is straightforward to generalize Lemmas~\ref{lem:unbiasedness} and \ref{lem:MSE} to this case. We omit the proofs for brevity.

\begin{lemma}[Unbiasedness] 
\label{lem:unbiasedness_3} 
The encoder $\alpha$ defined in \eqref{eq:randomized_protocol_3} is unbiased. That is,  $\EE{\alpha}{\alpha(X_i)} = X_i$ for all $i$. As a result, $Y$ is an unbiased estimate of the true average: $\EE{\alpha}{Y} = X$.
\end{lemma}

\begin{lemma}[Mean Squared Error]
\label{lem:MSE_3}
Let $\alpha = \alpha \left( p'_{ij}, p''_{ij}, \bar{X}'_i, \bar{X}''_i\right) $ be the protocol defined in \eqref{eq:randomized_protocol_3}. Then
\begin{equation*}
MSE_{\alpha}(X_1, \dots, X_n) = \frac{1}{n^2} \sum_{i=1}^n \sum_{j=1}^d \left( p'_{ij} \left( X_i(j) - \bar X'_i \right)^2 + p''_{ij} \left( X_i(j) - \bar X''_i \right)^2 + \left( p'_{ij} \bar X'_i + p''_{ij} \bar X''_i \right)^2 \right).
\end{equation*}
\end{lemma}

We expect the $k$-ary protocol to lead to better (lower) MSE bounds, but at the expense of an increase in communication cost. Whether or not the trade-off offered by $k>2$ is better than that for the $k=2$ case investigated in this paper is an interesting question to consider.


\subsection{Preprocessing via Random Rotations}

Following the idea  proposed in \cite{Distributed_mean}, one can explore an encoding protocol $\alpha_Q$ which arises as the composition of a random rotation, $Q$, applied to $X_i$ for all $i$, followed by the protocol $\alpha$ described in Section~\ref{sec:encode}.  Letting $Z_i = Q X_i$ and $Z=\frac{1}{n}\sum_i Z_i$, we thus have
\[Y_i = \alpha(Z_i), \qquad i=1,2,\dots,n.\]
With this protocol we associate the decoder
$\gamma(Y_1,\dots,Y_n) = \frac{1}{n}\sum_{i=1}^n Q^{-1} Y_i .$

Note that
\begin{eqnarray*}
MSE_{\alpha, \gamma}&=& \Exp{\left\|\gamma(Y_1,\dots,Y_n) - X\right\|^2} \\
&=& \Exp{\left\|Q^{-1}\gamma(Y_1,\dots,Y_n) - Q^{-1}Z\right\|^2} \\
&=& \Exp{\left\|\gamma(\alpha(Z_1),\dots,\alpha(Z_n)) - Z\right\|^2} \\
&=& \Exp{\Exp{\left\|\gamma(\alpha(Z_1),\dots,\alpha(Z_n)) - Z\right\|^2\;|\; Q}}.
\end{eqnarray*}

This approach is motivated by the following observation: a random rotation can be identified by a single random seed, which is easy to  communicate to the server without the need to communicate all floating point entries defining $Q$. So, a random rotation pre-processing step implies only a minor  communication overhead. However, {\em if} the preprocessing step helps to dramatically reduce the MSE, we get an improvement. Note that the inner expectation above is the formula for MSE of our basic encoding-decoding protocol, given that the data is $Z_i = QX_i$ instead of $\{X_i\}$. The outer expectation is over $Q$. Hence,  we would like the to find a mapping $Q$ which tends to transform the data $\{X_i\}$ into new data $\{Z_i\}$ with better MSE, in expectation.

From now on, for simplicity assume the node centers are set to the average, i.e.,  $\bar{Z}_i = \frac{1}{d}\sum_{j=1}^d Z_i(j)$. For any vector $x\in \R^d$, define
$$ \sigma(x) \eqdef \sum_{j=1}^d (x(j) - \bar{x})^2 = \|x-\bar{x}1\|^2, $$
where $\bar{x} = \tfrac{1}{d}\sum_j x(j)$ and $1$ is the vector of all ones. Further, for simplicity assume that $p_{ij}=p$ for all $i,j$. Then using Lemma~\ref{lem:MSE}, we get
$$ MSE = \frac{1-p}{pn^2}  \sum_{i=1}^n \EE{Q}{ \|Z_i - \bar{Z}_i 1\|^2} =  \frac{1-p}{pn^2}  \sum_{i=1}^n \EE{Q}{\sigma(Q X_i)}. $$

It is interesting to investigate whether   choosing $Q$ as a random rotation, rather than identity (which is the implicit choice done in previous sections), leads to improvement in MSE, i.e., whether we can in some well-defined sense obtain an inequality of the type
\[ \sum_i \EE{Q}{\sigma(Q  X_i)} \ll \sum_i \sigma(X_i).\]

This is the case for the quantization protocol proposed in \cite{Distributed_mean}, which arises as a special case of our more general protocol. This is because the quantization protocol is suboptimal within our family of encoders. Indeed, as we have shown, with a different choice of the parameter we  can obtained results which improve, in theory, on the rotation + quantization approach. This suggests that perhaps combining an appropriately chosen rotation pre-processing step with our optimal encoder, it may be possible to achieve further improvements in MSE for any fixed communication budget.  Finding suitable random rotations $Q$ requires  a careful study which we leave to future research.

\bibliographystyle{plain}
\bibliography{references}

\appendix

\section{Additional Proofs}
\label{sec:app:alternative_protocol}

In this section we provide proofs of Lemmas~\ref{lem:unbiasedness_2} and \ref{lem:MSE_2}, describing properties of the encoding protocol $\alpha$ defined in \eqref{eq:randomized_protocol_2}. For completeness, we also repeat the statements.

\begin{lemma}[Unbiasedness]
The encoder $\alpha$ defined in \eqref{eq:randomized_protocol} is unbiased. That is,  $\EE{\alpha}{\alpha(X_i)} = X_i$ for all $i$. As a result, $Y$ is an unbiased estimate of the true average: $\EE{\alpha}{Y} = X$.
\end{lemma}

\begin{proof}
Since $Y(j) = \frac{1}{n}\sum_{i=1}^n Y_{i}(j)$ and $X(j) = \frac{1}{n}\sum_{i=1}^n X_{i}(j)$, it suffices to show that $\EE{\alpha}{Y_i(j)}=X_i(j)$:
\begin{align*}
\EE{\alpha}{Y_i(j)} &= \frac{1}{|\sigma_k(d)|} \sum_{\sigma \in \sigma_k(d)} \left[ 1_{(j \in \sigma)} \left( \frac{d X_{i}(j)}{k} - \frac{d-k}{k} \mu_i \right) + 1_{(j \not\in \sigma)} \mu_i \right] \\
&= \binom{d}{k}^{-1} \left[ \binom{d-1}{k-1} \left( \frac{d X_{i}(j)}{k} - \frac{d-k}{k} \mu_i \right) + \binom{d-1}{k} \mu_i \right] \\
&= \binom{d}{k}^{-1} \left[ \binom{d-1}{k-1} \frac{d}{k} X_{i}(j) + \left( \binom{d-1}{k} - \binom{d-1}{k-1} \frac{d-k}{k} \right) \mu_i \right] \\
&= X_i(j)
\end{align*}
and the claim is proved.
\end{proof}

\begin{lemma}[Mean Squared Error]
Let $\alpha = \alpha(k)$ be encoder defined as in \eqref{eq:randomized_protocol_2}. Then
\begin{equation*}
MSE_{\alpha}(X_1, \dots, X_n) = \frac{1}{n^2} \sum_{i=1}^n \sum_{j=1}^d \frac{d-k}{k} \left( X_i(j) - \mu_i \right)^2.
\end{equation*}
\end{lemma}

\begin{proof} 
Using Lemma~\ref{lem:general_MSE}, we have
\begin{eqnarray}
MSE_{\alpha}(X_1,\dots,X_n) &=& \frac{1}{n^2} \sum_{i=1}^n \EE{\alpha}{\left\|Y_i - X_i \right\|^2} \notag \\
&=&\frac{1}{n^2} \sum_{i=1}^n \EE{\alpha}{\sum_{j=1}^d (Y_i(j) - X_i(j) )^2} \notag \\
&=&\frac{1}{n^2} \sum_{i=1}^n \sum_{j=1}^d \EE{\alpha}{ (Y_i(j) - X_i(j) )^2}.
\label{eq:6867sgs7_v2}
\end{eqnarray}
Further, 
\begin{align*}
\EE{\alpha}{ (Y_i(j) - X_i(j) )^2} &= 
\binom{d}{k}^{-1} \sum_{\sigma \in \sigma_k(d)} \left[ 1_{(j \in \sigma)} \left( \frac{d X_{i}(j)}{k} - \frac{d-k}{k} \mu_i - X_i(j) \right)^2 + 1_{(j \not\in \sigma)} \left( \mu_i - X_i(j) \right)^2 \right] \\
&= \binom{d}{k}^{-1} \left[ \binom{d-1}{k-1} \frac{(d-k)^2}{k^2} \left( X_i(j) - \mu_i \right)^2 + \binom{d-1}{k} \left( \mu_i - X_i(j) \right)^2 \right] \\
&= \frac{d-k}{k} \left( X_i(j) - \mu_i \right)^2.
\end{align*}
It suffices to substitute the above into \eqref{eq:6867sgs7_v2}.
\end{proof}

\end{document}